\documentclass[conference]{IEEEtran}
\usepackage[utf8]{inputenc}
\usepackage{amsmath,amsthm,amssymb}
\usepackage{float}
\usepackage{xcolor}
\usepackage{graphicx}
\usepackage{tikz}
\usepackage{algorithm}
\usepackage{algpseudocode}
\usepackage{amsmath,amsthm,amssymb}
\usepackage{mathtools} 
\usepackage{graphicx}
\usepackage{float}
\usepackage{xcolor}
\usepackage{tikz}
\usepackage{color,soul}
\usepackage{bbm}
\usetikzlibrary{patterns}
\newcommand{\xtic}[1]{(#1,0) -- ++(0,-0.25)}

\usepackage{cite}
\usepackage{hyperref}
\newcommand{\Geomv}{\mathsf{Geom}}
\usepackage{comment}
\usepackage{texdef2020}

\DeclareMathOperator*{\argmin}{arg\,min}
\newtheorem{theorem}{Theorem}
\newtheorem{definition}{Definition}

\newtheorem{lemma}{Lemma}
\def\mbf{\mathbf}

\definecolor{OI-vermillion}{RGB}{213,94,0}

\title{Timely Offloading in Mobile Edge Cloud Systems}

\author{%
 \IEEEauthorblockN{Nitya Sathyavageeswaran, Roy D. Yates, Anand D. Sarwate, Narayan Mandayam}
 \IEEEauthorblockA{Rutgers, The State University of New Jersey \\
                   Email: \{nitya.s, anand.sarwate\}@rutgers.edu, \{ryates, narayan\}@winlab.rutgers.edu}
}

\begin{document}
\maketitle

\begin{abstract}
Future real-time applications like smart cities will use complex Machine Learning (ML) models for a variety of tasks. Timely status information is required for these applications to be reliable. Offloading computation to a mobile edge cloud (MEC) can reduce the completion time of these tasks. However, using the MEC may come at a cost such as related to use of a cloud service or privacy. In this paper, we consider a source that generates time-stamped status updates for delivery to a monitor after processing by the mobile device or MEC. We study how a scheduler must forward these updates to achieve timely updates at the monitor but also limit MEC usage. We measure timeliness at the monitor using the age of information (AoI) metric. We formulate this problem as an infinite horizon Markov decision process (MDP) with an average cost criterion. We prove that an optimal scheduling policy has an age-threshold structure that depends on how long an update has been in service. 
\end{abstract}



\section{Introduction}

The Mobile Edge Cloud (MEC) is an emerging paradigm for computing in mobile environments. Mobile devices may be able to offload computation to a physically close edge cloud server for processing. This can help enable future real-time applications in smart cities~\cite{SmartCity}, autonomous vehicles~\cite{maqueda2018event}, healthcare~\cite{miotto2018deep}, and industrial monitoring~\cite{nasir2021review} which use complex machine learning (ML) models. For example, autonomous vehicles run ML algorithms to identify pedestrians, traffic signs, objects and other vehicles, and  predict their behaviour to provide a safe driving experience. Low latency is required for these applications so computing the output of the ML model should be done in a timely manner. 

A ML model may compute an output more quickly for some inputs than for others. A mobile device therefore has a choice of whether to offload the evaluation of the model to a MEC or to process it locally. We propose a model for this decision process in which a source (the device) sends time-stamped status updates (the ML model outputs) a monitor. We use the age of information (AoI) metric~\cite{kaul2012real} to measure the timeliness of these updates. The device, or \emph{scheduler} can choose to process the update locally or offload it, at a cost, to the MEC. Using the MEC can reduce the AoI, yielding a timeliness/cost tradeoff.


The scheduling problem in wireless systems from an AoI, privacy, and energy perspective has been extensively studied. Minimizing the long-run average age from a scheduling perspective for multiple users has been studied~\cite{ kadota2018optimizing, sun2018age, Hsu2020Scheduling}. 
Finding optimal policies to minimize age using Markov decision theory has been investigated~\cite{hsu2018age, CeranHARQ, tang2020minimizing, bedewy2021optimal}. From a privacy perspective, He et~al.~\cite{he2017privacy} address the privacy concerns associated with using the MEC and provide privacy aware-task scheduling algorithms in MEC systems. Min et~al.~\cite{min2018learning} propose offloading schemes to preserve privacy and save energy consumption of IoT devices. Task offloading with an energy constraint has been studied~\cite{mao2016dynamic, xu2016online}.  There has also been a lot of work to find optimal policies to minimize AoI subject to an energy constraint under different system settings~\cite{RoyLazyisTimely, BacinogluAgeEnergy, WuAgeEnergy, FengAgeEnergy, Abd-ElmagidAgeEnergy, stamatakis2019control}.

In this work, we study the tradeoff between the average age at the monitor and the frequency of MEC usage. We consider a generate-at-will source~\cite{RoyLazyisTimely} that generates time-stamped system updates which are sent to the monitor after processing by either the local server or MEC. 
We formulate the problem as a Markov decision process (MDP) with an average cost criterion, and prove that an optimal policy has an age-threshold structure depending on the amount of time an update has been in service. We also do simulations to compare the optimal policy to a few heuristic polices.

\section{System Model and Problem Formulation}

\noindent \textbf{Notation.} Random variables will be denoted by capital letters and realizations by lower case letters. 
Let  $\Geomv(\mu)$ denote the geometric distribution with mean $\mu$: if $X \sim \Geomv(\mu)$ then $\prob{X = k} = \mu (1- \mu)^{k-1}$, $k=1,2,\ldots$.

We study a discrete time system consisting of a source, scheduler, local server, mobile edge cloud (MEC) and a monitor as shown in Figure \ref{fig:sysmodel1}. 
Time passes in integer slots with slot $n\ge0$ denoting the time interval $[n,n+1)$. Within a slot $n$, the following events happen in order:
    \begin{enumerate}
        \item The source may or may not generate an update. The variable $x_n = 1$ if there is an update in slot $n$ and $x_n = 0$ otherwise.
        \item The scheduler observes $x_n$. If $x_n = 1$, the scheduler decides whether to forward the update to the local server or the MEC. If the update is sent to the local server, it enters a first-in-first-out (FIFO) queue to await service. If the update is sent to the MEC, it enters service at the MEC and is serviced immediately.
        \item The local server processes the head-of-line update in its queue according to its service distribution. 
        \item If the MEC or local server have finished servicing an update, it is sent to the monitor.
        \item The monitor may or may not request a new update from the source.
    \end{enumerate}
In this paper, we study a specific instance of this more general setup. We consider a source that has information about the service facility state. To avoid queuing delays, we consider a generate-at-will source that can submit a fresh update only after the previous update is serviced and received at the monitor. The local server has service distribution $\Geomv(\mu)$ and the MEC services an update in $1$ time slot. We assume that the transmission time to submit a packet to the MEC is much smaller than the computational time at the MEC, and is thus neglected.

\begin{figure}
    \centering
    \includegraphics[width=0.4\textwidth]{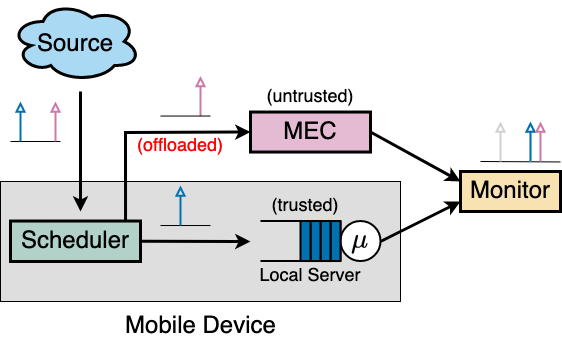}
    \caption{System model for offloading: Source generates status updates to a monitor after processing by the local server or MEC.}
    \label{fig:sysmodel1}
\end{figure}

\subsection{Scheduler-controlled source}

We allow the scheduler to abort local processing and ``pull'' a fresh update from the source at the start of the time slot. Thus, if an update is being processed by the local server and has been in service for a while, the scheduler can terminate/abort the local processing and instead pull a fresh update and sent it to the MEC. We model the scheduler's actions over time as a sequence $\mbf{u} = (u_1, u_2, \ldots)$, where 
    \begin{align}
    u_n&=\begin{cases}
        1 &\text{fresh update is sent to MEC}\\
        0 &\text{last update is processed by the local server}.
        \eqnlabel{actionscheduler}
    \end{cases}
    \end{align}
Once the local server or MEC finishes servicing an update, it is sent to the monitor.


\subsection{Problem Formulation}

We use AoI as a metric for timeliness. The AoI is defined as the amount of time that has elapsed since the generation of the most recent update that has been received at the monitor. Let $R_n = \max \{ k \le n : \text{update $x_k = 1$ received at monitor} \}$ so $A_n = n - R_n$ is the age at the monitor.
Because we will model the time for local processing as a random variable, the update times and age at the monitor are random processes. 


Our goal is to minimize the average age, defined as\footnote{Actions occur slot by slot in discrete time but age is averaged continuously over the slot; when $A_n=a$, the average age over slot $n$ is $a+1/2$.}
\begin{align}
    \Delta = \limty{N} \frac{1}{N}\sum_{n=1}^{N} (A_n +1/2) \eqnlabel{sum-age}.
    \end{align}
To disincentivize offloading every update, we impose a price for using the MEC. Since $u_n = 1$ when the scheduler offloads, the frequency of using the MEC is
	\begin{align}
	\
	\bar{p} = \lim_{N \to \infty} \frac{1}{N} \sum_{n=1}^{N} U_n.
	\label{eq:avgMECuse}
	\end{align}
This price can represent a number of different factors: a privacy loss for revealing the update to the MEC, the bandwidth/energy for transmitting the update to the MEC, or a monetary cost if the MEC is owned by an external service.

Given these ingredients we can formulate our problem as a Markov decision process (MDP)~\cite{puterman1994markov}. If $A_n = a_n$ is the age at the monitor at the end of slot $n$ and $z_n$ is the amount of service an update has received in slot $n$, the age $a_{n+1}$ and $z_{n+1}$ in slot $n+1$ can be computed as a function of $a_n$ and $z_n$. We therefore define the state of the system as $s_n=(a_n,z_n)$. The MDP  consists of a tuple ($\mathcal S$, $\mathcal{U}$, $P$, $C$): $s=(a,z)\in \mathcal{S}$, where $a\in \mathbb{N}$ and $z\in \mathbb{N}_0$ are the countably infinite set of states.
$u\in \mathcal{U}$ is the action space defined in \eqnref{actionscheduler}. The transition function $P$ governing the probability of the transition $s_{n} \to s_{n+1}$ is:
\begin{align}
    s_{n+1}=\begin{cases}
        (1,0) & \text{w.p. $1$ if $u_n = 1$} \\
        (z_n+1,0)&\text{w.p. $\mu$ if $u_n = 0$} \\
        (a_n+1, z_n+1) &\text{w.p. $1-\mu$ if $u_n = 0$}
    \end{cases}
    \eqnlabel{state:transition}.
\end{align}

We model the cost $C$ of taking action $u$ in state $s$ as the sum of the age at the monitor and the MEC cost:
\begin{align}
C(s_n,u_n) = a_n +1/2 + \lambda p_n,
\end{align}
where $\lambda \in (0,\infty)$ is a relative weight. 
A policy $f$ is a sequence of maps $\{ f_n \colon n \in \mathbb{N} \}$, where $f_n \colon (a_n,z_n)\to u_n$. 
We want to find a policy $f^*$ that minimizes the long term average cost: 
\begin{align}
        V_{f}(s)= \limsup \limits_{N\to \infty}\frac{1}{N}\Eop_{f}\bracket{ \sum_{n=1}^{N}  C(S_n,U_n) \middle\vert 
        S_0 = s  }. \eqnlabel{avgcost}
\end{align}



\section{Heuristic Policies}

We first analyze a few heuristic policies to understand the age/cost tradeoff.

\subsection{Local Server Only and MEC Only }

If we ignore the MEC and use only the local server, $\bar{p} = 0$. This is a discrete time, geometric service zero-wait policy. Using the discrete time version of the age analysis by Yates~\cite{RoyLazyisTimely}, the average age at the monitor equals
\begin{align}
    \Delta&= \frac{4-\mu}{2\mu}\eqnlabel{zerowaitgeo}.
\end{align}
In this policy, as the service rate $\mu$ approaches zero, the age at the monitor blows up. We thus need to look at other policies which use the MEC to reduce the age at the monitor. 

A policy that only uses the MEC to process the updates is the same using the local server when $\mu=1$. Hence it follows from  \eqnref{zerowaitgeo} that 
the average age is $\Delta = 1.5$. For this policy $\bar{p} = 1$, so while it achieves the lowest possible age at the monitor, it may be undesirable if the MEC is costly.


\subsection{Service Threshold Policy}

This policy is defined by the service time threshold $z^*$. If the amount of time an update has spent in service at the start of the slot $n$ is $z_n= z^*$, 
then the existing update in service is dropped and the source immediately generates a new update for the MEC to process. We find the average age at the monitor for this policy using the graphical approach
~\cite{kaul2012real}. Figure \ref{SampleServiceTimeThreshold} shows a sample variation of the age process at the monitor for the service time threshold policy. The source submits updates at times $U_1$, $U_2$,$\ldots$, $U_m$. The age increases linearly with rate $1$ in the absence of an  update at the monitor and drops to a smaller value $Y_i$ when an update finishes service and is received at the monitor. $Y_i$ is the age at the monitor when an update completes service and is defined as follows:
\begin{align}
    Y_i&=\begin{cases}
        Z_{i-1} & \text{if }Z_{i-1}\le z^*\\
        1 & \text{if } Z_{i-1}>z^*.
    \end{cases}
\end{align}

\begin{figure}[t]
\centering
\begin{tikzpicture}[scale=0.2]
\draw [fill=lightgray, ultra thin, dashed] (1,0) to (1,2) to (10,11) to (10,0) ;
\draw [fill=lightgray, ultra thin, dashed] (10,0) to (10,1) to (14,5) to (14,0) ;
\draw [fill=lightgray, ultra thin, dashed] (22,0) to (22,2) to (26,6) to (26,0) ;
\draw [fill=lightgray, ultra thin, dashed] (26,0) to (26,4) to (31,9) to (31,0) ;
\draw [<-|] (-2,16) node [above] {$\Delta(n)$} -- (-2,0) -- (16,0);
\draw [|->] (18,0) -- (36,0) node [right] {$n$};
\draw [thin]\xtic{-1}\xtic{0}\xtic{1}\xtic{2}\xtic{3}\xtic{4}\xtic{5}\xtic{6}\xtic{7}
\xtic{8}\xtic{9}\xtic{10}\xtic{11}\xtic{12}\xtic{13}\xtic{14}\xtic{19}\xtic{20}\xtic{21}\xtic{22}\xtic{23}\xtic{24}
\xtic{25}\xtic{26}\xtic{27}\xtic{28}\xtic{29}\xtic{30}\xtic{31}\xtic{32}\xtic{33}\xtic{34}; 

\draw [<->] (1,1.8) -- node [right] {$Y_1$} (1,0.4);
\draw [<->] (26,3.8) -- node [right] {$Y_m$} (26,0.4);
\draw (-0.5,1.5) node {$\Tilde{Q}_0$};
\draw[<-] (5.5,1.5) to [out=110,in=250] (5.5,11) node [above] {$Q_1$};
\draw[<-] (12,1.5) to [out=110,in=250] (12,11) node [above] {$Q_2$};
\draw[<-] (24,1.5) to [out=110,in=250] (24,11) node [above]{$Q_{m-1}$};
\draw[<-] (29.5,3.5) to [out=110,in=250] (29.5,11) node [above]{$Q_{m}$};
\draw [very thick] (-2,3) -- (1,6) -- (1,2) -- (10,11) -- (10,1)  -- (14,5) -- (14,4) -- (16,6);
\draw [very thick] (22,2)  -- (26,6) -- (26,4) -- (31,9) -- (31,3) -- (33,5); 

\draw  (1,0) node {$\bullet$} node [below] {$U_1$};
\draw  (10,0) node {$\bullet$} node [below] {$U_2$};
\draw  (14,0)node {$\bullet$} node [below] {$U_3$};
\draw  (22,0) node {$\bullet$} node [below] {$U_{m-1}$};
\draw  (26,0)  node {$\bullet$} node [below] {$U_m$};
\draw [thin] (31,0.4) -- (31,-0.4);
\draw  (31,0) node [below] {${\mathcal{T}_m}$};
\draw  [|<->|] (1,-3) to node [below] {$z^*+1$} (10,-3);
\draw  [|<->|] (10,-3) to node [below] {$Z_2=4$} (14,-3);
\draw  [|<->|] (22,-3) to node [below] {$Z_{m-1}$} (26,-3);
\draw  [|<->|] (26,-3) to node [below] {$Z_m$} (31,-3);
\draw  [|<->|] (1,-6) to node [below] {$\hat{S}_1=9$} (10,-6);
\draw  [|<->|] (10,-6) to node [below] {$\hat{S}_2=4$} (14,-6);
\draw  [|<->|] (22,-6) to node [below] {$\hat{S}_{m-1}$} (26,-6);
\draw  [|<->|] (26,-6) to node [below] {$\hat{S}_m$} (31,-6);
\end{tikzpicture}
	\caption{Sample variation of the age process at the monitor with $z^*=8$.} 
\label{SampleServiceTimeThreshold}
\end{figure}
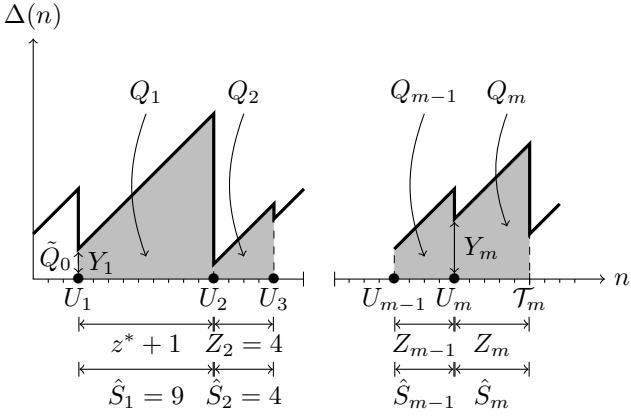
 Let $Z_i$ denote the service times for the local server. Since either the local server or the MEC can service an update, we define the effective service times $\hat{S_i}$ as follows:
\begin{align}
    \hat{S_i}&=\begin{cases}
        Z_i & \text{if }Z_i\le z^*\\
        z^*+1 &\text{if } Z_i>z^*.
    \end{cases} \eqnlabel{effservice}
\end{align}
Let $m$ denote the number of updates that finish service in time $\mathcal{T}_m$.
The average age at the monitor is the area under the saw tooth function in Figure \ref{SampleServiceTimeThreshold}
normalized by the total time of observation $\mathcal{T}_m$ given by
\begin{align}
    \mathcal{T}_m&= \hat{S}_1+ \hat{S}_2+\ldots+\hat{S}_m \eqnlabel{Tmsum}.
\end{align}The average age at the monitor is given by 
\begin{align}
    \Delta&= \frac{\E{Q_i}}{\Esmall{\hat{S}_i}}=\frac{\Esmall{Y_i\hat{S}_i}+\frac{1}{2}\Esmall{{\hat{S}_i}^2}}{\Esmall{\hat{S}_i}} \eqnlabel{areacal}.
\end{align}
Since $Y_i$ and $\hat{S}_i$ are independent,
\begin{align}
    \Delta&= \E{Y}+ \frac{1}{2}\frac{\Esmall{{\hat{S}}^2}}{\Esmall{\hat{S}}}\eqnlabel{expcalc}\\
    &= \frac{2\left(1-\bar{\mu}^{z^*}(\mu z^*+\bar{\mu})-\bar{\mu}^{z^*+1}+ \bar{\mu}^{2z^*+1}(\mu z^*+\bar{\mu})\right)}{2\mu\left(1-{\bar{\mu}}^{z^*+1}\right)}\nonumber\\
    &\qquad + \frac{\mu^2\bar{\mu}^{z^*}(z^*+1) +\mu- {\bar{\mu}}^{z^*}z^*- \mu {\bar{\mu}}^{z^*}}{2\mu\left(1-{\bar{\mu}}^{z^*+1}\right)} \nonumber\\
    &\qquad + \frac{2\bar{\mu}-{\bar{\mu}}^{z^*+1}(2+z^*\mu )}{2\mu\left(1-{\bar{\mu}}^{z^*+1}\right)}\eqnlabel{avgagest}.
\end{align}
Computation of the expectation of $Y$, $\hat{S}$ and $\hat{S}^2$ is straightforward but tedious and details on computing \eqnref{areacal} and \eqnref{expcalc} are deferred to the Appendix~\ref{Appendixa}.

The frequency of MEC usage is
\begin{align}
    \bar{p}&= \lim _{\mathcal{T}_m\to\infty}\frac{\sum_z \mathbbm{1}\{Z_i>z^*\}}{\mathcal{T}_m}\\
    &=\lim _{\mathcal{T}_m\to\infty} \frac{\sum_ z \mathbbm{1}\{Z_i>z^*\}/m}{\mathcal{T}_m/m}=\frac{\prob{ Z_i>z^*} }{ \E{\hat{S}} }\\
     &=\frac{\mu {\bar{\mu}}^{z^*}}{1-{\bar{\mu}}^{z^*+1}}.
\end{align}

\subsection{Age Threshold Policy}
%
 This policy is defined by the age threshold $a^*$: 
\begin{align}
    u_n&=
    \begin{cases}
        1 & \text {if }a_n=a^*\\
        0& \text{otherwise}.
    \end{cases}
\end{align}
Thus when $a_n=a^*$, the existing update in service is dropped and the source immediately generates a new update for the MEC to process. The MEC processes this update in slot $n$ so the age at the start of the next slot is $a_{n+1}=1$.

We formulate the $(A,Z)$ Markov chain for this policy to find $\Delta$ and $\bar{p}$. 
Figure~\ref{fig:samplemarkovchain} shows a Markov chain for the age threshold policy with $a^*=3$.
\begin{figure}
    \centering
    \includegraphics[width=0.35\textwidth]{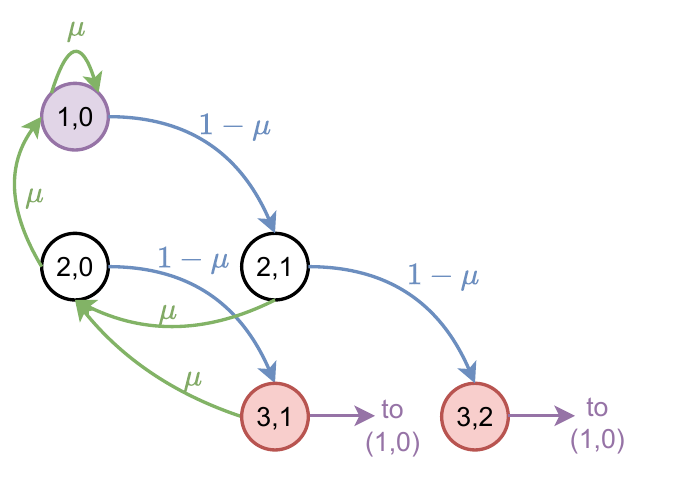}
    \caption{Markov chain for the age threshold policy with $a^*=3$. }
    \label{fig:samplemarkovchain}
\end{figure}
We find the stationary distribution $\pi(a,z)$ of the Markov chain numerically and use that to find $\E{A}=\sum_{a,z}a\pi(a,z)$.
The average age at the monitor for the continuous age function is then given by
\begin{align}
    \Delta = \E{A} + 1/
2.\end{align}
Similarly we can also find the MEC frequency, $\bar{p}=\sum_{z} \pi(a^*,z)$ from the stationary distribution $\pi(a,z)$ of the  Markov chain. 

\section{Age Threshold Based Optimal Policy}

We now turn to finding the optimal policy $f^*$ which minimizes the long term average cost in \eqnref{avgcost}. We first define the discounted cost $V_{f,\beta}(s)$ as follows:
    \begin{align}
    V_{f,\beta}(s)= \Eop_{f} \bracket{ \sum_{n=1}^\infty \beta^n C(s_n,u_n)
        \middle\vert 
        S_0 = s }, 
    \eqnlabel{discountedsennot}
    \end{align} 
where $s_0=(1,0)$ is the initial system state and $\beta\in(0,1)$ is the discounting factor. We define:
\begin{align}
    V_\beta(s)= \inf _{f} V_{f,\beta}(s).
\end{align}
If the value function $V_{\beta}(s)$ is finite, we have that $V_\beta(s)$ satisfies the following discounted cost optimality equation:
\begin{align}
    V_\beta(s)&= 
\min_{u\in\mathcal{U}} \Bigl(C(s,u)+ \beta \sum_{s'} P_{ss'}(u) V_\beta (s')\Bigr), \eqnlabel{disccost}
\end{align}
for all $s\in \mathcal{S}$ where $P_{ss'}(u)$ is the transition probability from state $s$ to $s'$ under action $u$. 
We can define a value iteration $V_{\beta,n}(s)$ recursively as:
\begin{align}
     &V_{\beta,n}(s) \nonumber\\
     &= \min_{u} \Bigl(C(s,u)+ \beta \sum_{s'} P_{ss'}(u) V_{\beta,n-1} (s')\Bigr)\\
    & = a+ 1/2+ \min \Bigl(\lambda + \beta V_{\beta, n-1}(1,0),\Bigr.\nonumber\\ &\left. \beta\mu V_{\beta, n-1}(z+1,0) + \beta (1-\mu) V_{\beta, n-1}(a+1, z+1)\right),\eqnlabel{valuefunction}
\end{align} for any $n\ge 1$ and $V_{\beta,0}(s)=0$ for all $s\in \mathcal{S}$ and $\beta\in (0,1)$.

Sennot's results show that for a finite $V_\beta(s)$ satisfying certain conditions~\cite[Assumptions~1-$3^*$]{sennott1989average} , 
there exists a differential cost function $h(s)$ and a stationary deterministic policy $f^*$ that satisfies the optimality equation 
\begin{align}
        g+ h(s)= \min_{u} \Bigl(C(s,u)+ \sum_{s'}P_{ss'}(u) h(s')\Bigr), 
\end{align}
for all $s\in \mathcal{S}$, where  $g$ is the average cost by following policy $f^*$, and $s'$ is the resulting  state after taking action $u$.  Details showing that Sennot's conditions, such as $V_\beta (s)$ being finite, hold for our problem are in Appendix~\ref{Appendixb}.

\subsection{Structure of the Optimal Policy}
\begin{definition}
    A $z$-dependent age threshold policy is a stationary deterministic policy of the MDP in which if the action of the policy for state $s=(a,z)$ is to use the MEC, then the action for state $s'=(a+1,z)$ is also to use the MEC. 
\end{definition}
For a $z$-dependent age threshold policy, for each $z$ there is an age threshold $\bar{a}_z$ given by
\begin{align}
    \bar{a}_z = \min \{a \colon  u(a,z) = 1\}\eqnlabel{abarz},
\end{align} for which we use the MEC.

In the theorem below  we show that the optimal policy $f^*$ for the MDP  has a $z$- dependent age threshold structure.

\begin{theorem}\label{agethresholdtheorem}
    There exists a $z$-dependent age threshold policy of the MDP that minimizes the long term average cost. 
\end{theorem}
\begin{proof}
    The proof technique is similar to that by Hsu et~al.~\cite{hsu2017age}. We start by showing the age-threshold policy is optimal for the discounted cost case given by \eqnref{discountedsennot} and generalize this for  the average cost case. 

    Consider the action of using the MEC in state $s=(a,z)$ to be optimal. Then from \eqnref{disccost} we have
    \begin{align}
        \beta \mu V_\beta(z+1,0)+ \beta (1-\mu) V_\beta(a+1,z+1)&\ge \lambda + \beta V_\beta (1,0)\eqnlabel{soptimal}. 
    \end{align}
    Now we need to show that the optimal action for state $s'=(a+1,z)$ is also to use the MEC. Hence we need to show
    \begin{align}
        \beta \mu V_\beta(z+1,0)+ \beta (1-\mu) V_\beta(a+2,z+1)&\geq \lambda + \beta V_\beta (1,0)\eqnlabel{sprimeoptimal}.
    \end{align}
    We have 
    \begin{align}
     &\beta \mu V_\beta(z+1,0)+ \beta (1-\mu) V_\beta(a+2,z+1)\nonumber\\
     &\qquad \geq \beta \mu V_\beta(z+1,0)+ \beta (1-\mu) V_\beta(a+1,z+1)\eqnlabel{monotonicityn}\\
    &\qquad\ge\lambda + \beta V_\beta (1,0) \eqnlabel{valmonn},
\end{align}
where \eqnref{monotonicityn} follows from the monotonicity property of the value function, and \eqnref{valmonn} follows from \eqnref{soptimal}.
This concludes the proof for the discounted case. 

We now need to generalize the structure of the optimal policy for the average cost case. Let $\{\beta_k\}_{k=1}^\infty$ be a sequence of  discount factors converging to $1$. Following Sennot~\cite{sennott1989average}, the optimal policy for minimizing the total  $\beta_k$-discounted cost converges to the policy $f^*_k$ for the average cost case. 
\end{proof}
\begin{lemma} \label{threshlemma1}
    If the action of the policy for state $s=(a,z)$ is to use the MEC, then the action for state $s=(a,z+1)$ is also to use the MEC. Also,  if the action of the policy for state $s=(a,z)$ is to use the local server, then the action for state $s=(a,z-1)$ is also to use the local server. 
\end{lemma}
The proof is similar to that of Theorem~\ref{agethresholdtheorem} and is in Appendix~\ref{Appendixc}
From Lemma \ref{threshlemma1} we have 
\begin{align}
    \bar{a}_0\ge \bar{a}_1\ge \bar{a}_2\ge \ldots 
\end{align}
This is intuitive because as $z$ increases, we get a better age reduction by using the MEC and thus $\bar{a}_z$ is monotonically non-decreasing in $z$. 
\subsection{Relative Value Iteration}

To find the optimal policy $f^*$, we want to apply the relative value iteration (RVI) method~\cite{puterman1994markov}. 
But our state space is countably infinite and the cost is unbounded. We therefore truncate our state space as follows: if the age of an update goes above some specified tolerance level $a_{\textrm{max}}$, the scheduler then requests a new update from the source and sends it to the MEC. The state of the system then goes from $(a_{\textrm{max}},z)$ to $(1,0)$. The optimal policy for the truncated MDP and the original MDP are asymptotically identical when $a_{\textrm{max}}\to \infty$ if the truncated MDP is \textit{unichain}~\cite{puterman1994markov}. Our truncated MDP is unichain since any stationary policy $f$ for that chain has only one recurrent class, as the state $(1,0)$ (recurrent state) can be reached from all other states $(a,z)$. Under any policy, two consecutive service completions has non zero probability following which we are at state $(1,0)$.  

We therefore apply the RVI algorithm to the finite state approximation of our problem. In each iteration of the RVI algorithm we compute the value function for each state as follows:
\begin{align}
    V_{n}(s)&= \min_{u} \left[C(s,u)+ \sum_{s'} P_{ss'}(u) V_{n-1} (s') - V_{n-1}(1,0)\right]
\end{align}
where, $V_0(s)=0$ for all $s\in \mathcal{S}$. 

\begin{algorithm}
\caption{RVI algorithm given $\mu$ and $\lambda$}\label{algo1}
\hspace*{\algorithmicindent} \textbf{Input:} $\mu$, $\lambda$\\
\hspace*{\algorithmicindent} \textbf{Initialize:} $V(s)\gets 0$ for all states $s\in \mathcal{S}$;
\begin{algorithmic}[1]
\While{$1$}
\ForAll {$s\in \mathcal{S}$}
\If {$a=a_{\mathrm{max}}$}
\State $u\gets 1 $;
\Else
\State $u\gets \argmin_{u\in\mathcal{U}} (C(s,u) + \E{V(s')})$;
\EndIf
\State $V_{\mathrm{temp}}(s)\gets \min_{u} \left(C(s,u) + \E{V(s')}\right) - V(1,0)$;
\EndFor
\State  $V_{\mathrm{temp}}(s)\gets V(s)$ for all $s\in \mathcal{S}$.
\EndWhile\label{algo1}
\end{algorithmic}
\end{algorithm}

\begin{figure}
    \centering
    \includegraphics[width=0.35\textwidth]{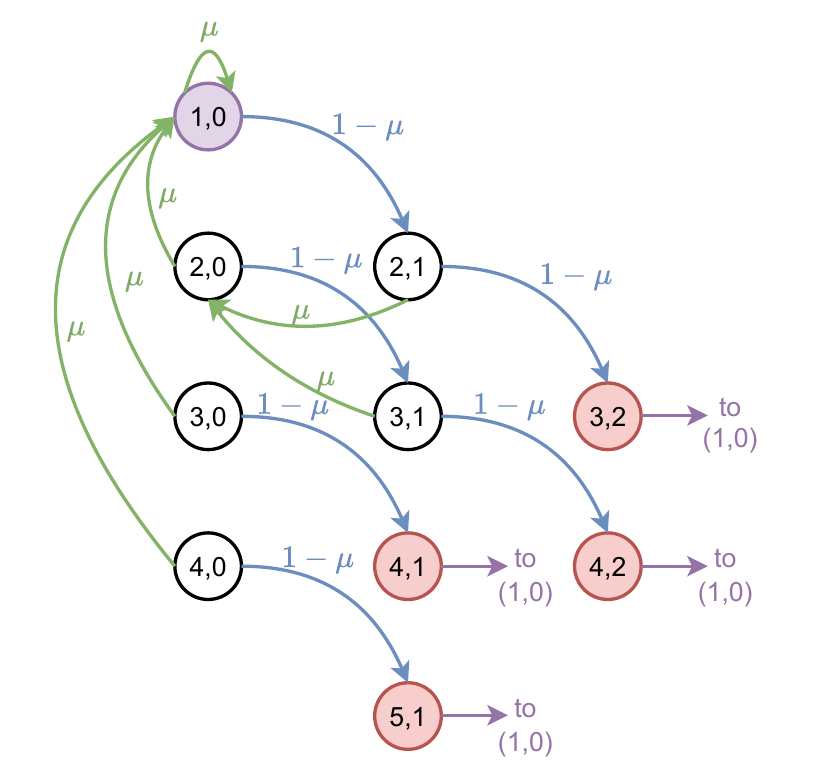}
    \caption{Markov chain of the optimal policy for $\mu=0.5$ and $\lambda=3$. The green lines  indicate  local service completion. The blue lines indicate that the update is still in service in the local server. The red circles indicate MEC service. $\bar{a}_1=4$, $\bar{a}_2=3. $
}
    \label{fig:optimalpolicy}
\end{figure}
Figure~\ref{fig:optimalpolicy} illustrates the optimal policy when the local server service rate, $\mu=0.5$, and $\lambda=3$. From Figure~\ref{fig:optimalpolicy} we see that the local processor keeps working on an update till the age reaches $\bar{a}_z$. 
Once the age equals $\bar{a}_z$, 
the MEC takes over and services a new update that it receives from the source immediately. 
We see that the age threshold is monotonically non-decreasing in $z$. In Figure \ref{fig:optimalpolicy}, 
$\bar{a}_1=4$ and $\bar{a}_2=3$.

\subsection{Simulations}

Figure~\ref{agemec} shows the variation of age with the frequency of MEC usage when local service rate, $\mu=0.01$. For the service threshold and age threshold policies, as $z^*$ or $a^*$ decreases respectively, we use the MEC more often, and the age decreases. 
As $a^*$, $z^*$, $\lambda$ approach infinity, the age blows up for each of the policies and is equivalent to only using the local server. For $(\bar{p}, \Delta)=(1,1.5)$  all policies are equivalent to only using the MEC in every slot. 
The service threshold policy gives us a better age- MEC frequency tradeoff than the age threshold policy. In fact, the service threshold policy is very close to the optimal policy. 
This is a useful observation since the service threshold policy is easier to implement in practice than the optimal policy, as we need to compute the relative values for each state $(a,z)$ recursively to find the optimal policy.
\begin{figure}[t]
    \centering
    \includegraphics[width=0.5\textwidth]{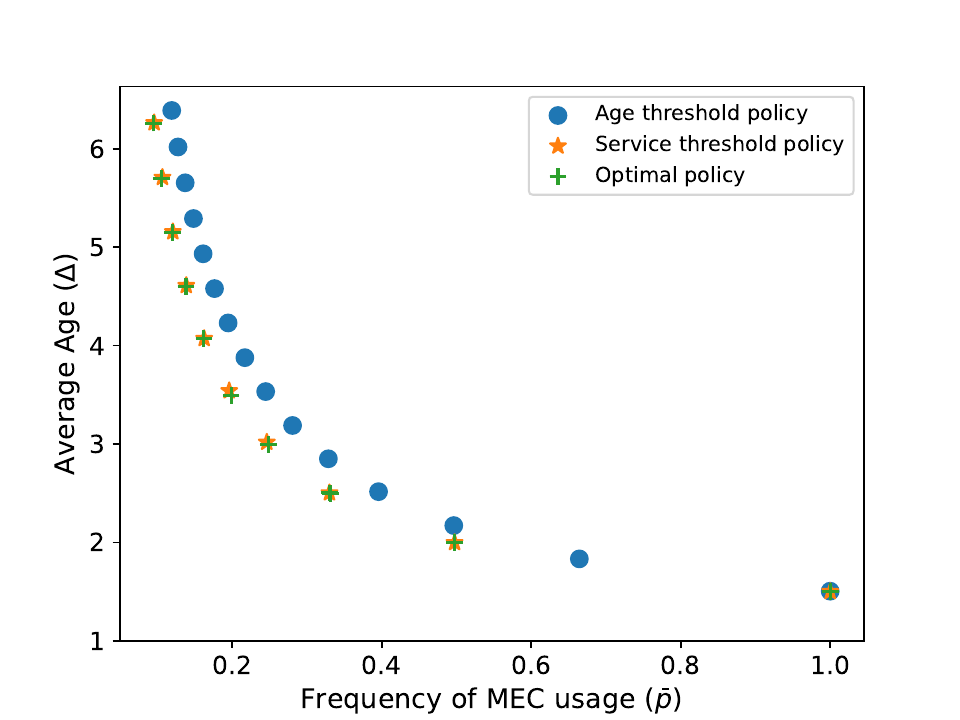}
    \caption{Age vs.~frequency of using the MEC for  $\mu=0.01$. $a^*$ varies from 1 to 15 for the age threshold policy, and $z^*$ varies from 0 to 9 for the service threshold policy. For the optimal policy, $\lambda\in (0,50)$. }
    \label{agemec}
\end{figure}

\section{Conclusions and Future Work}

In this paper we looked at a Mobile Edge Cloud System in which the MEC  helps a local device process updates in a timely way, but with a cost. We provided structural results for the optimal policy that minimizes the long term average cost: this policy applies an age threshold dependent on the amount of time the update has spent in service. Proving that the service threshold policy is a close approximation to the optimal policy is an interesting open question.

Additional future work includes: analyzing systems with multiple users, finding optimal policies (or the structure thereof) for more realistic modelling of the latency of the MEC which accounts for features in real systems such as round-trip time (RTT), multiple edge servers, or network issues such as contention. Ultimately, we hope to add quality of experience (QoE) considerations into our model for specific applications.


\section*{Acknowledgements}

This material is based upon work supported by the National Science Foundation under grant CNS-2148104 and is supported in part by funds from federal agency and industry partners as specified in the Resilient \& Intelligent NextG Systems (RINGS) program.

\clearpage

\IEEEtriggeratref{13}

\bibliographystyle{IEEEtran}
\bibliography{bibage.bib}

\newpage\clearpage

\section{Appendix}

\subsection{Service Threshold Age Analysis }\label{Appendixa}
The average age at the monitor is the area under the saw tooth function in Figure \ref{SampleServiceTimeThreshold}
normalized by the total time of observation $\mathcal{T}_m$.   The total area under the function is the sum of areas $\Tilde{Q}_0$, $Q_1$, \ldots, $Q_m$.
\begin{align}
    \Delta_{\mathcal{T}_m}&=\frac{Q_0+ \sum_{i=1}^m Q_i}{\mathcal{T}_m}\eqnlabel{avgstage}.
\end{align}
The area $Q_i$ can be calculated as the sum of the area of $2$ different entities:  
i) rectangle with height $Y_i$ and whose base connects the points $U_i$ and $U_{i+1}$, ii) isosceles triangle with width $\hat{S_i}$ and whose base connects the points $U_i$ and $U_{i+1}$. We then get
\begin{align}
    Q_i=Y_i\hat{S_i}+\frac{1}{2}\hat{S_i}^2 \eqnlabel{Qarea}.
\end{align}

Substituting \eqnref{Qarea} into \eqnref{avgstage} and rearranging some terms yields
\begin{align}
    \Delta_{\mathcal{T}_m}&= \frac{Q_0}{\mathcal{T}_m}+ \frac{m}{\mathcal{T}_m}\frac{\sum_{i=1}^m \left(Y_i\hat{S}_i+ \frac{1}{2}\hat{S_i}^2\right)}{m}.\eqnlabel{timeavgage}
\end{align}
As $\mathcal{T}_m\to \infty$, the first term in \eqnref{timeavgage} vanishes and we get
\begin{align}
    \lim_{\mathcal{T}_m\to \infty} \frac{m}{\mathcal{T}_m}= \frac{1}{\E{\hat{S}}}.
\end{align}
The summation term in \eqnref{timeavgage} is a sample average and it will converge to its stochastic average when $\mathcal{T}_m\to \infty$.
The average age at the monitor can now be calculated as follows:
\begin{align}
    \Delta= \lim_{\mathcal{T}_m\to \infty} \Delta_{\mathcal{T}_m} = \frac{\E{Y\hat{S}}+\frac{1}{2}\E{\hat{S}^2}}{\E{\hat{S}}}. \eqnlabel{avgageservice}
\end{align}

Since $Y_i$ and $\hat{S_i}$ are independent, it is sufficient for us to find $ \E{\hat{S}}$ and $\E{A}$. From \eqnref{effservice}, we find $ \E{\hat{S}}$ as follows
\begin{align}
    \E{\hat{S}}&=\E{Z_i|Z_i\le z^*}P(Z_i\le z^*)+ (z^*+1)P(Z_i>z^*). \eqnlabel{expshat}
\end{align}
Define, $\bar{\mu}=1-\mu$. Since $Z \sim \Geomv(\mu)$ we have
\begin{align}
    P(Z_i\le z^*)&= 1- {\bar{\mu}}^{z^*} \eqnlabel{eq1}\\
    P(Z_i> z^*)&= {\bar{\mu}}^{z^*}\eqnlabel{eq2}.
\end{align}
Going back to the first term in \eqnref{expshat} we have
\begin{align}
    &\E{Z_i|Z_i\le z^*}P(Z_i\le z^*)\nonumber\\
    &\qquad= \sum_{j=1}^{z^*} j P(Z_i=j|Z_i\le z^*)P(Z_i\le z^*)\\
    &\qquad= \sum_{j=1}^{z^*} j \frac{P(Z_i=j)}{P(Z_i\le z^*)}P(Z_i\le z^*)\\
    &\qquad=\sum_{j=1}^{z^*} j {\bar{\mu}}^{(j-1)}\mu\\
    &\qquad= \frac{1-{\bar{\mu}}^{z^*}}{\mu}- z^*{\bar{\mu}}^{z^*}\eqnlabel{ags},
\end{align}
where \eqnref{ags} follows from the sum of an arithmetico-geometric series. Substituting  \eqnref{eq1}, \eqnref{eq2} and \eqnref{ags} in \eqnref{expshat} we get
\begin{align}
    \E{\hat{S}}&= \frac{1-{\bar{\mu}}^{z^*}}{\mu}-z^*{\bar{\mu}}^{z^*}+ (z^*+1){\bar{\mu}}^{z^*}\\
    &=\frac{1-{\bar{\mu}}^{z^*+1}}{\mu} \eqnlabel{firstmoment}. 
\end{align}
Similarly we find $\E{Y}$ as follows
\begin{align}
    \E{Y}&=\E{Z_{i-1}|Z_{i-1}\le z^*}P(Z_{i-1}\le z^*)+ P(Z_{i-1}>z^*)\\
    &= \frac{1- {\bar{\mu}}^{z^*}(\mu z^*+\bar{\mu})}{\mu} \eqnlabel{aexp}.
\end{align}
We also have
\begin{align}
    \E{\hat{S}^2}&= \E{{Z_i}^2|Z_i\le z^*}P(Z_i\le z^*)+ {(z^*+1)}^2P(Z_i>z^*).\eqnlabel{expsquare}
\end{align}
Looking at the first term in \eqnref{expsquare} we have
\begin{align}
    &\E{{Z_i}^2|Z_i\le z^*}P(Z_i\le z^*)\nonumber\\
    &\qquad= \sum_{j=1}^{z^*} j^2 P(Z_i=j|Z_i\le z^*)P(Z_i\le z^*)\\
    &\qquad=\sum_{j=1}^{z^*} j^2 \frac{P(Z_i=j)}{P(Z_i\le z^*)}P(Z_i\le z^*)\\
        &\qquad=\sum_{j=1}^ {z^*} j^2 \bar{\mu}^{j-1} \mu \\
    &\qquad= \frac{2\bar{\mu}-\mu z^*{\bar{\mu}}^{z^*}(1+z^*\mu)-{\bar{\mu}}^{z^*+1}(2+z^*\mu)}{\mu^2} \nonumber\\
    &\qquad\qquad + \frac{1-{\bar{\mu}}^{z^*}}{\mu}- z^*{\bar{\mu}}^{z^*} \eqnlabel{expzsquared}.
\end{align}

Substituting \eqnref{expzsquared} and \eqnref{eq2} into \eqnref{expsquare}, we get
\begin{align}
    \E{{\hat{S}}^2}
    &=\frac{2\bar{\mu} -{\bar{\mu}}^{z^*+1}(2+z^*\mu )}{\mu^2} +\frac{1-{\bar{\mu}}^{z^*}z^*- {\bar{\mu}}^{z^*}}{\mu}\nonumber\\
    &\qquad + {\bar{\mu}}^{z^*}z^*+ {\bar{\mu}}^{z^*}.\eqnlabel{secondmoment}
\end{align}


\subsection{Proving Assumptions of Sennot~\cite{sennott1989average}}\label{Appendixb}
\begin{lemma}\label{Lemma1}
     For every state $s$ and discount factor $\beta$, the quantity $V_{\beta}(s)$ is finite.
\end{lemma}
\begin{proof}
    From Sennot~\cite{sennott1989average}, it is sufficient to show the MDP  has a stationary policy $f$ inducing an irreducible, ergodic Markov chain satisfying the following condition:
\begin{align}
    \sum _{s} \pi_{s} C(s,u)<\infty,
\end{align}
where $\pi_s$ is the steady-state distribution of the Markov chain.

Consider the following stationary policy: Use the MEC for all states $s$. The Markov chain reduces to a single state $(1,0)$ having steady state probability of $1$. 
Then 
\begin{align}
    \sum_{s} \pi_s C(s,u)&= c(1,0)= 1+ \lambda< \infty.
\end{align}
\end{proof}

We use \eqnref{minfact} to help prove Lemmas~\ref{Lemma2} and \ref{Lemma3}. For any functions $f'$,$g'$, $f$ and $g $, if $f'\ge f$ and $g'\ge g$ then 
\begin{align}
    \min(f',g')-\min(f,g)\ge 0. \eqnlabel{minfact}
\end{align}

Using $\eqnref{valuefunction}$ we show the following $2$ lemmas about $V_\beta(s)$. 
\begin{lemma}[Monotonicity of Value function]\label{Lemma2}
    For fixed $z$, $V_{\beta, n}(a,z)$ is monotonically non-decreasing in $a$.
\end{lemma}

\begin{proof}
    We prove this lemma by induction. For $n=1$, we have 
    \begin{align}
        V_{\beta, 1}(a,z)=0
    \end{align}
    for all $a\in \mathbb{N}$ and $z\in \mathbb{N}_{0}$.
    Thus the lemma holds for $n=1$. 
    Now let us say the lemma holds for some $n=k$. Then,
\begin{align}
    V_{\beta, k}(a,z)\le V_{\beta, k}(a+1,z)
\end{align}
for all $a\in \mathbb{N}$ and $z\in \mathbb{N}_{0}$.
Now we need to show the claim holds for $n=k+1$.  We have
\begin{align}
    &V_{\beta, k+1}(a+1,z)- V_{\beta, k+1}(a,z)\nonumber\\
    &\qquad=a+ 1/2+ 1 +  \min \left(\lambda + \beta V_{\beta, k}(1,0),\right.\nonumber\\
    &\qquad\left.\beta\mu V_{\beta, k}(z+1,0) + \beta (1-\mu) V_{\beta, k}(a+2, z+1)\right)\nonumber\\
    &\qquad\qquad - a- 1/2 - \min \left(\lambda + \beta V_{\beta, k}(1,0),\right.\nonumber\\
    &\qquad\left.\beta\mu V_{\beta, k}(z+1,0) + \beta (1-\mu) V_{\beta, k}(a+1, z+1)\right), \eqnlabel{valuefunctiondiff}
\end{align}
for all $a\in\mathbb{N}$ and $z\in \mathbb{N}_0$.
 
Let
\begin{align}
    f'&=\lambda + \beta V_{\beta, k}(1,0)\\
    f&=\lambda + \beta V_{\beta, k}(1,0)\\
    g'&=\beta\mu V_{\beta, k}(z+1,0) + \beta (1-\mu) V_{\beta, k}(a+2, z+1)\\
    g&=\beta\mu V_{\beta, k}(z+1,0) + \beta (1-\mu) V_{\beta, k}(a+1, z+1).
\end{align}
We see that $f'=f$ and
\begin{align}
    &g'-g\nonumber\\
    &= 1 + \beta (1-\mu) [V_{\beta, k}(a+2, z+1)-V_{\beta, k}(a+1, z+1)].
\end{align}
From the induction assumption we know $V_{\beta, k}(a+2, z+1)-V_{\beta, k}(a+1, z+1)\ge 0$. Thus, we have $f'=f$ and $g'\ge g$. Going back to \eqnref{valuefunctiondiff}, we have
\begin{align}
    &V_{\beta, k+1}(a+1,z)- V_{\beta, k+1}(a,z)\nonumber\\
    &\qquad=1 + \min (f',g')- \min(f,g)\geq 0 \eqnlabel{fact1}
\end{align}
for all $a\in \mathbb{N}$ and $z\in \mathbb{N}_0$. The inequality \eqnref{fact1} follows from \eqnref{minfact}. 

Thus the claim holds for $n=k+1$, completing the proof of Lemma \ref{Lemma2}.
\end{proof}

\begin{lemma}[Monotonicity of Value function]\label{Lemma3}
    For fixed $a$, $V_{\beta, n}(a,z)$ is monotonically non-decreasing in $z$.
\end{lemma}
\begin{proof}
    We prove this lemma by induction. For $n=1$, we have 
    \begin{align}
        V_{\beta, 1}(a,z)= 0
    \end{align}
    for all $a\in \mathbb{N}$ and $z\in \mathbb{N}_{0}$.
    Thus the lemma holds for $n=1$. 
    Now let us say the lemma holds for some $n=k$. Then,
\begin{align}
    V_{\beta, k}(a,z)\le V_{\beta, k}(a,z+1)
\end{align}
for all $a\in \mathbb{N}$ and $z\in \mathbb{N}_{0}$.
Now we need to show the claim holds for $n=k+1$.  We have
\begin{align}
    &V_{\beta, k+1}(a,z+1)- V_{\beta,k+1}(a,z)\nonumber\\
    &\qquad=a+1/2+\min \left(\lambda + \beta V_{\beta, k}(1,0),\right.\nonumber\\
    &\qquad\left.\beta\mu V_{\beta, k}(z+2,0) + \beta (1-\mu) V_{\beta, k}(a+1, z+2)\right)\nonumber\\
    &\qquad\qquad - a -1/2- \min \left(\lambda + \beta V_{\beta, k}(1,0),\right.\nonumber\\
    &\qquad\left.\beta\mu V_{\beta, k}(z+1,0) + \beta (1-\mu) V_{\beta, k}(a+1, z+1)\right),
\end{align}
for all $a\in \mathbb{N}$ and $z\in \mathbb{N}_0$.
Let
\begin{align}
    f'&=\lambda + \beta V_{\beta, k}(1,0)\\
    f&=\lambda + \beta V_{\beta, k}(1,0)\\
    g'&=\beta\mu V_{\beta, k}(z+2,0) + \beta (1-\mu) V_{\beta, k}(a+1, z+2)\\
    g&=\beta\mu V_{\beta, k}(z+1,0) + \beta (1-\mu) V_{\beta, k}(a+1, z+1).
\end{align}
We see that
\begin{align}
    &f'-f= \beta (V_{\beta, k}(z+2,0)-V_{\beta, k}(z+1,0))\\
    &g'-g= \beta \mu (V_{\beta, k}(z+2,0)-V_{\beta, k}(z+1,0)) \nonumber \\
    &\qquad +\beta (1-\mu) [V_{\beta, k}(a+1, z+2)-V_{\beta, k}(a+1, z+1)].
\end{align}
 From Lemma \ref{Lemma2} we have $V_{\beta, k}(z+2,0)-V_{\beta, k}(z+1,0)\ge 0$. From the induction assumption we have  $V_{\beta, k}(a+1, z+2)-V_{\beta, k}(a+ 1, z+1)\ge 0$. Thus, we have  $f'\ge f$ and $g'\ge g$. Going back to \eqnref{valuefunctiondiff}, we have
\begin{align}
    V_{\beta, k+1}(a,z+1)- V_{\beta, k+1}(a,z)&= \min(f',g')- \min(f,g)\\
    &\geq 0 \eqnlabel{fact2}
\end{align}
for all $a\in \mathbb{N}$ and $z\in \mathbb{N}_0$. Inequality \eqnref{fact2} follows from $\eqnref{minfact}$. 

Thus the claim holds for $n=k+1$, completing the proof of Lemma \ref{Lemma3}.
\end{proof}
Let $h_\beta (s)= V_\beta (s)- V_{\beta}(s_0)$.
\begin{lemma}\label{Lemma 4}
       There exists a non-negative $N$ such that $-N\le h_\beta (s) $ for all  $s\in \mathcal{S}$ and $\beta$. 
\end{lemma}
\begin{proof}
    For each $\beta$, we have 
\begin{align}
    V_{\beta, n}(a,z)\uparrow V_\beta(a,z) \eqnlabel{VnapproachV},
\end{align}
for every $a\in \mathbb{N}$ and $z\in \mathbb{N}_0$. 
From Lemma \ref{Lemma3} and \eqnref{VnapproachV}, we have
\begin{align}
    V_\beta (a,z)- V_\beta (a,0)\geq 0 \eqnlabel{ineq1}
\end{align} and from Lemma \ref{Lemma2} and \eqnref{VnapproachV},  we have
\begin{align}
    V_\beta(a,0)- V_{\beta}(1,0)\geq 0. \eqnlabel{ineq2}
\end{align}
Adding \eqnref{ineq1} and \eqnref{ineq2}, we get
\begin{align}
    V_\beta(a,z)- V_\beta(1,0)\geq 0. \eqnlabel{ass2proof}
\end{align}
From \eqnref{ass2proof}, we see that the Lemma holds for $N=0$. 
\end{proof}
\begin{lemma} \label{Assumption 3}
     There exists a non-negative $M(s)$, such that $h_\beta (s)\le M(s)$ for every $s$ and $\beta$. For every $s$ there exists an action $u(s)$ such that $\sum _{s'} P_{ss'}(u(s)) M(s')<\infty$. In addition, $\sum_{s'} P_{ss'}(u) M(s')<\infty$ for all $s\in \mathcal{S}$ and $u\in \mathcal{U}$.
\end{lemma}
\begin{proof}
    The proof is same as that of Lemma\ref{Lemma1}~\cite{sennott1989average}. 
\end{proof}
\subsection{Proof of Lemma\ref{threshlemma1}}\label{Appendixc}
    We start by showing the lemma is true for the discounted cost case given by \eqnref{discountedsennot} and  then generalize this for the average cost case. 

    Let's consider the action of using the MEC in state $s=(a,z)$ to be optimal. Then we have from \eqnref{disccost}
    \begin{align}
        \beta \mu V_\beta(z+1,0)+ \beta (1-\mu) V_\beta(a+1,z+1)&\ge \lambda + \beta V_\beta (1,0) \eqnlabel{soptimal1}. 
    \end{align}
    Now we need to show that the optimal action for state $s'=(a,z+1)$ is also to use the MEC. Hence we need to show
    \begin{align}
        \beta \mu V_\beta(z+2,0)+ \beta (1-\mu) V_\beta(a+1,z+2)&\geq \lambda + \beta V_\beta (1,0)\eqnlabel{sprimeoptimal1}.
    \end{align}
    We have
    \begin{align}
     &\beta \mu V_\beta(z+2,0)+ \beta (1-\mu) V_\beta(a+1,z+2)\nonumber\\
     &\qquad\geq \beta \mu V_\beta(z+1,0)+ \beta (1-\mu) V_\beta(a+1,z+1)\eqnlabel{monotonicity1}\\
    &\qquad\ge\lambda + \beta V_\beta (1,0) \eqnlabel{valmon1},
\end{align}
where \eqnref{monotonicity1} follows from Lemma \ref{Lemma1} and Lemma \ref{Lemma2}, and \eqnref{valmon1} follows from \eqnref{soptimal1}.
This concludes the proof for the discounted case. 

We now need to generalize the structure of the optimal policy for the average cost case. Let $\{\beta_k\}_{k=1}^\infty$ be a sequence of  discount factors converging to $1$. From Sennot~\cite{sennott1989average}, the optimal policy for minimizing the total  $\beta_k$-discounted cost converges to the policy $f^*_k$ for the average cost case. 
The proof for the second statement in Lemma~\ref{threshlemma1} is similar, and therefore omit it. 

\end{document}